\documentclass[12pt]{amsart}
\usepackage{amssymb,latexsym}
\usepackage{amsfonts}
\usepackage{amsmath}
\usepackage[colorlinks,linkcolor=blue,anchorcolor=blue,citecolor=blue]{hyperref}
\usepackage{algorithm}
\usepackage{enumerate}
\usepackage{algpseudocode}
\usepackage{verbatim}
\usepackage{graphicx}
\usepackage[all]{xy}
\usepackage{amssymb}
\usepackage{slashbox}

\newcommand{\Z}{{\mathbf{Z}}}

\newcommand{\cS}{{\mathcal{S}}}

\newcommand{\gZ}{{\mathrm Z}}

\newcommand{\F}{{\mathbb F}}

\newcommand{\cC}{{\mathcal C}}
\newcommand{\cD}{{\mathcal D}}

\newtheorem{theorem}{Theorem}[section]
\newtheorem{lemma}[theorem]{Lemma}

\theoremstyle{definition}
\newtheorem{definition}[theorem]{Definition}

\theoremstyle{remark}
\newtheorem{remark}[theorem]{Remark}

\numberwithin{equation}{section}

\def\BibTeX{{\rm B\kern-.05em{\sc i\kern-.025em b}\kern-.08em
    T\kern-.1667em\lower.7ex\hbox{E}\kern-.125emX}}

\begin{document}

\title{On the Constructions  of MDS Self-dual Codes via Cyclotomy}

\author{Aixian Zhang}
\address{Department of Mathematical Sciences, Xi'an  University of Technology,
Shanxi, 710054, China.}
\email{zhangaixian1008@126.com}

\author{Keqin Feng}
\address{Department of Mathematical
Sciences, Tsinghua University,
 Beijing, 100084, China. }
\email{fengkq@tsinghua.edu.cn}





\begin{abstract}
MDS self-dual codes over finite fields have attracted a lot of attention in recent
years by their theoretical interests in coding theory and applications in cryptography
and combinatorics. In this paper we present a series of MDS self-dual codes with new
length by using generalized Reed-Solomon codes and extended generalized Reed-Solomon codes
as the candidates of MDS codes and taking their evaluation sets as an union of cyclotomic classes.
The conditions on such MDS codes being self-dual are expressed in terms of cyclotomic numbers.
\end{abstract}

\maketitle



\section{Introduction}\label{sec-one}

Let $\F_q$ be the finite field with $q$ elements. An $[n,k,d]_q$ linear code $\cC$
is a $k$-dimensional subspace of $\F^n_q$ with minimum (Hamming) distance $d.$
It is well known that $n, k, d$ need satisfy the Singleton bound $d \leq n-k+1.$
If the equality is attained then the code is called MDS code.
The dual code $\cC^{\perp}$ of $\cC$ is defined by
$$
\cC^{\perp}=\{ v \in \F^n_q: (v,c)=0 \ \mbox{for all} \ c \in \cC \}
$$
where for $v=(v_1, v_2,\ldots,v_n)$ and $c=(c_1, c_2,\ldots,c_n),$ $(v,c)=\sum^n\limits_{i=1}v_i c_i \in \F_q $
is the Euclidean inner product in $\F^n_q$. The code $\cC$ is called self-dual if $\cC=\cC^{\perp}.$
If $\cC$ is both MDS and self-dual, $\cC$ is called MDS self-dual code.

MDS codes and self-dual codes are important families of classical codes in coding theory.
 Therefore, it is of interests to investigate MDS self-dual codes. Since the dimension and
 distance are determined by the length of an MDS self-dual code, thus we usually focus on the
 length and the field size of MDS self-dual codes. The problem is completely solved by
Grassl and Gulliver \cite{GG} when $q$ is even, but not for odd $q$. One of the basic problems
of this topic is to determine the existence of MDS self-dual codes for length $n$ and a fixed finite field $\F_q$.

In the literature, there are many known constructions of MDS self-dual codes.
One of the method to construct MDS self-dual codes is based on the generalized Reed-Solomon (GRS for short) codes or
 extended generalized Reed-Solomon (EGRS for short) codes
 \cite{FF}, \cite{FLLL}, \cite{JX}, \cite{Yan}, \cite{ZF}, since they are MDS codes.
The codewords of GRS codes and EGRS codes are made by evaluation of polynomials in
$\F_q[x]$ at $\cS$ for a certain subset $\cS$ of the projective line
$\F_q \bigcup \{ \infty\}.$ The conditions on $\cS$ have been provided in order that
such MDS codes constructed with $\cS$ are self-dual by Jin and
Xing \cite{JX} for GRS case $(\cS  \subseteq \F_q)$
and Yan \cite{Yan} for EGRS case $(\infty \in \cS)$ respectively. In most of previous works,
the set $\cS$ is chosen as a union of cosets of a subgroup of $\F^{\ast}_q$ or a subspace
of $\F_q.$

In this paper, we consider the construction of MDS self-dual codes over $\F_q$ by using the
first approach. Namely, we take $\cS$ as a union of cosets of a subgroup of $\F^{\ast}_q$.
Let  $\F^{\ast}_q=\langle \theta \rangle$ where $\theta$ is a primitive element of $\F^{\ast}_q$.
Any subgroup of $\F^{\ast}_q$ is a cyclic group $\cD=\langle \theta^e \rangle $ where $q-1=ef,
|\cD|=f$ and all cosets of $\cD$ in $\F^{\ast}_q$ are the $e$-th cyclotomic classes
$$
\cD^{(e)}_i=\theta^i \cD= \{ \theta^{i+e \lambda}: 0 \leq \lambda \leq f-1\}
\quad ( 0 \leq i \leq e-1), \ \ \cD=\cD^{(e)}_0.
$$

In the previous works \cite{FF}, \cite{FLLL}, \cite{LLL}, \cite{ZF}, $q$ is a square,
$q=r^2,r=p^m \ (p \geq 3)$ and $\cS$ is a union of $\cD^{(e)}_i$ with several particular $i$
satisfying $r-1 \mid i.$ In this paper, we consider the case $q,$ which is any prime power
and we take $\cS$ being a union of cosets in more flexible way, so that we get many new series
of MDS self-dual codes with length $n$ . For doing this we use the properties and computations
on cyclotomic numbers.

This paper is organized as follows. In Section \ref{sec-two}, we introduce the
basic results given in \cite{JX} and \cite{Yan} on criteria of MDS self-dual codes
constructed by GRS and EGRS codes being self-dual.
We also introduce the basic properties of cyclotomic numbers in Section \ref{sec-two}
which are main machinary of this paper. In Section \ref{sec-three}, we present our
general results on constructing MDS self-dual codes over $\F_{q}$ by using cyclotomic classes
of $\F^{\ast}_q$. Then we show several particular cases as applications of our general
results in Section \ref{sec-four}.

\section{Preliminaries}\label{sec-two}

\subsection{ MDS Self-dual codes Constructed by GRS Codes and EGRS codes}

In this subsection, we briefly review some basic results on GRS codes
and EGRS codes. For the details, the reader may refer to \cite{HP} and \cite{MS}.

\begin{definition}
Let $ q=p^m, \cS=\{a_1, a_2, \ldots, a_n\}$ be a subset of $\F_{q}$ with $n$
distinct elements, $v_1, v_2, \ldots, v_n$ be nonzero elements in $\F_q$
(not necessarily distinct), $v=(v_1, v_2, \ldots, v_n ).$ For $ 1 \leq k \leq n-1,$ the GRS code
is defined by
\begin{eqnarray*}
\cC_{grs}(\cS,v,q)=\{c_f &=& (v_1f(a_1),v_2f(a_2),\ldots,v_nf(a_n)) \in \F^n_q: \\ && f(x) \in \F_q[x], \deg f \leq k-1\}.
\end{eqnarray*}
This is an MDS (linear) code over $\F_q$ with parameters $[n,k,d]_q, d=n-k+1.$
The extended GRS code is defined by
\begin{eqnarray*}
\cC_{egrs}(\cS,v,q)=\{c_f&=&(v_1f(a_1),v_2f(a_2),\ldots,v_nf(a_n),f_{k-1})
\in \F^{n+1}_q:\\ && f(x) \in \F_q[x], \deg f \leq k-1\}
\end{eqnarray*}
where $f_{k-1}$ is the coefficient of $x^{k-1}$ in $f(x).$
This is also an MDS (linear) code over $\F_q$ with parameters
$[n+1,k,d]_q, d=n-k+2.$
\end{definition}

For $\cS=\{a_1, a_2, \ldots, a_n\} \subseteq \F_q,$  we denote
$$
\Delta_{\cS}(a_i)=\prod_{\scriptstyle 1 \leq j \leq n
\atop \scriptstyle j \neq i}(a_i -a_j) \in \F^{\ast}_q.
$$
Let $\eta_q: \F^{\ast}_q \rightarrow \{\pm 1\}$ be the quadratic (multiplicative) character of $\F_q.$
Namely, for $b \in \F^{\ast}_q$,
$$
\eta_q(b) =  \left \{
\begin{array}{ll}
1, & \mbox{if} \ b \ \mbox{is a square in } \F^{\ast}_q, \\
-1, & \mbox{otherwise} .
\end{array}
\right.
$$
Then $\eta_q(b)=(-1)^{\varphi_q(b)}$ where $\varphi_q(b): \F^{\ast}_q \rightarrow \F_2$ is defined by
$$
\varphi_q(b) =  \left \{
\begin{array}{ll}
0, & \mbox{if} \ b \ \mbox{is a square in } \F^{\ast}_q, \\
1, & \mbox{otherwise} .
\end{array}
\right.
$$
Namely, if $\F^{\ast}_q=\langle \theta \rangle,$ then $\varphi_q(\theta^l) \equiv l \  (\bmod ~2) \ (1 \leq l \leq q-1).$

A sufficient condition on set $\cS$ has been given in \cite{JX}
and \cite{Yan} for $\mathcal{C}_{grs}$
and $\mathcal{C}_{egrs}$    being
self-dual. From the proofs we can see that the sufficient condition is also necessary.
Now we introduce the basic results given in \cite{JX} and \cite{Yan}.

\begin{theorem}\label{thm-main}
Let $a_1, a_2, \ldots, a_n$ be distinct elements in $\F_q, \cS=\{a_1, a_2,$ $ \ldots, a_n\}.$

(1)\ (\cite{JX}) Suppose that $n$ is even.
There exists $v=(v_1, v_2, \ldots, v_n ) \in (\F^{\ast}_q)^{n}$ such that the (MDS)
code $\cC_{grs}(\cS,v,q)$ is self-dual if and only if all $\eta_q(\Delta_{\cS}(a)) \ (a \in \cS)$
are the same (which means that all $\varphi_q(\Delta_{\cS}(a)) \ (a \in \cS)$ are the same).

(2)\ (\cite{Yan}) Suppose that $n$ is odd.
There exists $v=(v_1, v_2, \ldots, v_n )\in (\F^{\ast}_q)^{n}$ such that the (MDS) code
$\cC_{egrs}(\cS,v,q)$ is self-dual code  if and only if
$\eta_q (-\Delta_{\cS}(a))=1$ for all $a \in \cS$
(which means that $\varphi_q(-\Delta_{\cS}(a))=0 \ \mbox{for all} \  a \in \cS $).
\end{theorem}

\begin{definition}
Let $\Sigma(q)$ be the set of all even number $n \geq 2$
such that there exists MDS self-dual code over $\F_q$
with length $n$. Let $\Sigma(g,q)$ and $\Sigma(eg,q)$ be the set of all even number $n\geq 2$
such that there
exists MDS self-dual code over $\F_q$ with length $n$
constructed by GRS code (Theorem \ref{thm-main} (1))
and EGRS code (Theorem \ref{thm-main} (2)) respectively. Namely,
$$
\Sigma(g,q)=\left \{ n:      \begin{array}{ll}
2 \mid n \geq 2, \mbox{there exists a subset } \ \cS \ \mbox{of} \ \F_q, |\cS|=n,  \\
\mbox{such that all} \ \eta_q(\Delta_{\cS}(a)) \  (a \in \cS)\ \mbox{are the same}.
\end{array} \right\}.
$$

$$
\Sigma(eg,q)=\left\{ n:      \begin{array}{ll}
2 \mid n \geq 2, \mbox{there exists a subset } \ \cS \ \mbox{of} \ \F_q, |\cS|=n-1,  \\
\mbox{such that all} \ \eta_q(-\Delta_{\cS}(a))=1 \ \mbox{for all } \ a \in \cS.
\end{array}\right\}.
$$
\end{definition}
We have $\Sigma(g,q)  \cup \Sigma(eg,q) \subseteq \Sigma(q).$

\subsection{ Cyclotomic Numbers}

A brief background on cyclotomic numbers is given in the following.
For more details, the reader is referred to the book \cite{Storer}.

Let $q=p^m$ where $p$ is an odd prime, $m \geq 1.$
Let $q-1=ef, e \geq 2, \F^{\ast}_q=\langle \theta \rangle, \cD=\langle \theta^e \rangle .$
The cosets of the subgroup $\cD$ in $\F^{\ast}_q$ are the
following $e$-th cyclotomic classes
$$
\cD_{\lambda}=\cD^{(e)}_{\lambda}=\theta^{\lambda}\cD=\{ \theta^{\lambda+ej}: 0 \leq j \leq f-1\}  \ \ (0 \leq \lambda \leq e-1).
$$

\begin{definition}
For $0 \leq i,j \leq e-1,$ the $e$-th cyclotomic numbers for $\F^{\ast}_q=\langle \theta \rangle$
are defined by
$$
(i,j)=(i,j)_e=|(\cD_i +1) \cap \cD_j|=\sharp \{ x \in \cD_i : x+1 \in \cD_j \}.
$$
\end{definition}

\begin{lemma}\label{thm-cyclotomy}
Let $q=p^m$ where $p$ is an odd prime, $m \geq 1, q-1=ef$ and
$(i,j)=(i,j)_e \ ( 0 \leq i,j \leq e-1)$
be the $e$-th cyclotomic numbers for $\F^{\ast}_q=\langle \theta \rangle$.

(1) $(i,j)=(-i,j-i)=(pi,pj).$

(2) $
(i,j)=  \left \{
\begin{array}{ll}
(j,i), & \mbox{if} \ 2 \mid f, \\
(j+\frac{e}{2},i+\frac{e}{2} ), & \mbox{if}  \ 2 \nmid f.
\end{array}
\right.
$

(3) $\sum^{e-1}\limits_{j=0}(i,j)=f-\theta_{i}$, where
 $
\theta_i=  \left \{
\begin{array}{ll}
1, & \mbox{if} \ 2 \mid f, i=0 \ \mbox{or} \ 2 \nmid f, i=\frac{e}{2}, \\
0, &  \mbox{otherwise}.
\end{array}
\right.
$

\quad $\sum^{e-1}\limits_{i=0}(i,j)=f-\delta_{j,0}$, where
 $
\delta_{j,0}=  \left \{
\begin{array}{ll}
1, & \mbox{if} \ j=0, \\
0, &  \mbox{otherwise}.
\end{array}
\right.
$
\end{lemma}

In this paper, we are concerned with the cases of even number $e.$  The values of $e$-th cyclotomic
numbers for $e=2$ and $4$ are listed as follows.

\begin{lemma}(\cite{Storer})\label{thm-cyc}
Let $q=p^m, p \geq 3, q-1=ef, (i,j)=(i,j)_e $
be the $e$-th cyclotomic numbers of $\F^{\ast}_q$.

(1) For $e=2,$

(1.1) If $2 \mid f,$ then $(0,0)=\frac{f}{2}-1, (0,1)=(1,0)=(1,1)=\frac{f}{2};$

(1.2) If $2 \nmid f,$ then $(0,1)=\frac{f+1}{2}, (0,0)=(1,0)=(1,1)=\frac{f-1}{2}.$

(2) For $e=4,$ we have $q=s^2 +4t^2$ where $s \in \Z$ is determined by
$s \equiv 1 \  (\bmod ~4)$ and $t$ is determined up to sign.

(2.1) If $2 \mid f,$ the values of $(i,j)=(i,j)_4$ are listed in Table I where

$16A=q-11-6s, 16B=q-3+2s+8t, 16C=q-3+2s,16D=q-3+2s-8t, 16E=q+1-2s. $

\vspace{0.5cm}
\begin{minipage}{\textwidth}
 \begin{minipage}[t]{0.45\textwidth}
  \centering
\begin{tabular}{c|c c c c }
\multicolumn{4}{c}{Table I \ $e=4, \ 2\mid f$}\\
\backslashbox{i}{j}
 & 0 & 1 & 2 & 3 \\
\hline
0 & A & B & C & D \\
1 & B & D & E & E \\
2 & C & E & C & E \\
3 & D & E & E & B \\
\end{tabular}
\end{minipage}
  \begin{minipage}[t]{0.45\textwidth}
   \centering
\begin{tabular}{c|c c c c }
\multicolumn{4}{c}{Table II \ $e=4, \ 2\nmid f$}\\
\backslashbox{i}{j}
 & 0 & 1 & 2 & 3 \\
\hline
0 & A & B & C & D \\
1 & E & E & D & B \\
2 & A & E & A & E \\
3 & E & D & B & E \\
\end{tabular}
\end{minipage}
\end{minipage}

\vspace{0.5cm}
(2.2) If $2 \nmid f,$ the values of $(i,j)=(i,j)_4$ are listed in Table II where

$16A=q-7+2s, 16B=q+1+2s-8t, 16C=q+1-6s,16D=q+1+2s+8t, 16E=q-3-2s. $
\end{lemma}

\section{Main Results}\label{sec-three}

Let $q=p^m$ where $p$ is an odd prime and $m \geq 1,
\F^{\ast}_q=\langle \theta \rangle, q-1=ef, 2 \mid e,
\cD=\langle \theta^e \rangle,\cD_{\lambda}=\cD^{(e)}_{\lambda}
=\theta^{\lambda}\cD \  ( 0 \leq \lambda \leq e-1).$
For a subset $I$ of $\Z_e=\{0,1,\cdots,e-1\}, |I|=l \ ( 1 \leq l \leq e). $
Let $\cS$ and $\widetilde{\cS}$ be subsets of $\F_q$ defined by

\begin{equation}\label{eqn-two}
\cS=\bigcup_{\lambda \in I} \cD_{\lambda}, \quad \widetilde{\cS}=\cS \bigcup \{0\},
\end{equation}
then $|\cS|=fl, |\widetilde{\cS}|=fl+1.$

The following Lemma follows from the aforementioned Theorem \ref{thm-main}.
\begin{lemma}\label{thm-impor}
(1) Assume that $2 \mid fl.$ If $\varphi_q(\Delta_{\cS}(a)) \in \F_2 =\{0,1\} \ (a \in \cS)$
are the same, then $fl \in \Sigma(g,q).$ If $\varphi_q(\Delta_{\widetilde{\cS}}(a))=\varphi_q(-1)$
for all $ a \in \widetilde{\cS},$ then $fl+2 \in \Sigma(eg,q).$

(2) Assume that $2 \nmid fl.$ If $\varphi_q(\Delta_{\cS}(a)) = \varphi_q(-1)$  for all $a \in \cS,$
 then $fl+1 \in \Sigma(eg,q).$ If $\varphi_q(\Delta_{\widetilde{\cS}}(a)) \  (a \in \widetilde{\cS})$
are the same,  then $fl+1 \in \Sigma(g,q).$
\end{lemma}

Now we compute  $\varphi_q(\Delta_{\cS}(a))$ and
$\varphi_q(\Delta_{\widetilde{\cS}}(a))$
by using the $e$-th cyclotomic numbers $(i,j)=(i,j)_e$ on $\F_q.$ For two subsets $I,J$ of
$\Z_e=\{0,1,\cdots,e-1\}, |I|, |J| \geq l$. Denote
$$
(I,J)=\sum\limits_{\scriptstyle i \in I
\atop \scriptstyle j \in J} (i,j), \ (i,J)=(\{i\}, J), (I,j)=(I,\{j\})
$$
and
$$
(I, \mbox{odd})=\sum^{e-1}\limits_{\scriptstyle j=0
\atop \scriptstyle  2 \nmid j}(I,j), \ \ \  (I, \mbox{even})=\sum^{e-1}\limits_{\scriptstyle j=0
\atop \scriptstyle  2 \mid j}(I,j)
$$
$(\mbox{odd}, J)$ and $(\mbox{even}, J)$ can be defined similarly.

\begin{lemma}\label{thm-three}
Let $\cS$ and $\widetilde{\cS}$ be subsets of $\F_q$ defined by (\ref{eqn-two}).
Then for each $a \in \cD_i, i \in I,$
\begin{eqnarray*}
\varphi_q(\Delta_{\cS}(a)) & \equiv & (fl-1)(i+\frac{ef}{2})+(\mbox{odd},I-i) \  (\bmod ~2) \\
\varphi_q(\Delta_{\widetilde{\cS}}(a)) & \equiv & \varphi_q(\Delta_{\cS}(a))+i \  (\bmod ~2) \\
\varphi_q(\Delta_{\widetilde{\cS}}(0)) & \equiv & fl\frac{e}{2}+f |I_{\mbox{odd}}| \  (\bmod ~2)
\end{eqnarray*}
where $I_{\mbox{odd}}=\{ i \in I: 2 \nmid i\}, I-i=\{j-i: j \in I\}.$
\end{lemma}
\begin{proof}
For each $a \in \cD_i, i \in I,$
\begin{eqnarray*}
\Delta_{\cS}(a) &=& \prod\limits_{\scriptstyle b \in \cS
\atop \scriptstyle  b \neq a} (a-b)=\prod_{\lambda \in I}\prod\limits_{\scriptstyle b \in \cD_{\lambda}
\atop \scriptstyle  b \neq a}(a-b) \ \ \ \ (\mbox{let} \ b=ac) \\
&=& \prod_{\lambda \in I}\prod\limits_{\scriptstyle c \in \cD_{\lambda-i}
\atop \scriptstyle  c \neq 1}(a-ac)=(-a)^{fl-1}\prod_{\lambda \in I}\prod\limits_{\scriptstyle c \in \cD_{\lambda-i}
\atop \scriptstyle  c \neq 1}(c-1).
\end{eqnarray*}
Note $2 \mid e,$ we know that for $\xi \in \cD_{\lambda}, \varphi_q(\xi) \equiv  \lambda  \  (\bmod ~2)$
and $\varphi_q(-1)=\frac{ef}{2}.$ Hence
\begin{eqnarray*}
\varphi_q(\Delta_{\cS}(a)) & \equiv & (fl-1)(\frac{ef}{2}+i)+\prod_{\lambda \in I}
\prod\limits_{\scriptstyle c \in \cD_{\lambda-i}
\atop \scriptstyle  c-1  \in \cD_{\mu}, 2\nmid \mu } 1 \  (\bmod ~2)  \\
& \equiv & (fl-1)(\frac{ef}{2}+i)+(\mbox{odd}, I-i)\  (\bmod ~2).
\end{eqnarray*}

On the other hand, from $\widetilde{\cS}=\cS \bigcup \{0\},$ we get
$\Delta_{\widetilde{\cS}}(a)=\Delta_{\cS}(a)a.$ Thus
$$
\varphi_q(\Delta_{\widetilde{\cS}}(a)) \equiv  \varphi_q(\Delta_{\cS}(a))+i \  (\bmod ~2).
$$
At last, $\Delta_{\widetilde{\cS}}(0)=\prod\limits_{a \in \cS}(-a)=(-1)^{fl}\prod\limits_{i \in I}\prod\limits_{a \in \cD_i}a.$
Therefore
$$
\varphi_q(\Delta_{\widetilde{\cS}}(0)) \equiv  fl\frac{ef}{2}+f|I_{\mbox{odd}}|
\equiv fl\frac{e}{2}+f|I_{\mbox{odd}}|  \  (\bmod ~2).
$$
\end{proof}

The following theorem will play a central role in determining the existence of MDS self-dual codes.
\begin{theorem}\label{thm-six}
Let $q=p^m \  (p \geq 3),q-1=ef, 2 \mid e,\F^{\ast}_q=\langle \theta \rangle,
\cD_{\lambda}=\theta^{\lambda}\langle \theta \rangle \ (0 \leq \lambda \leq e-1).$
Let  $I$ be a subset of $\Z_e=\{0,1,\cdots,e-1\}, |I|=l, 1 \leq l \leq e,$
$\cS=\bigcup\limits_{\lambda \in I} \cD_{\lambda},  \widetilde{\cS}=\cS \cup \{0\},$
so that
$|\cS|=fl$ and $|\widetilde{\cS}|=fl+1.$ We get

\emph{ Case }1: $2 \mid f.$

\ (1.1) If $i+(\mbox{odd}, I-i) \  (\bmod ~2)$ are the same for all $i \in I,$
then $fl \in \Sigma(g,q).$

\ (1.2) If $(\mbox{odd}, I-i)$ are even for all $i \in I,$
then $fl+2 \in \Sigma(eg,q).$

\emph{ Case 2}: $2 \nmid f$ and $2 \mid l.$

\ (2.1) If $i+(\mbox{odd}, I-i) \  (\bmod ~2)$ are the same for all $i \in I,$
then $fl \in \Sigma(g,q).$

\ (2.2) If $ |I_{\mbox{odd}}| \equiv  \ \frac{e}{2} \  (\bmod ~2)$
and $(\mbox{odd}, I-i) \equiv 0 \  (\bmod ~2)$ for all $i,$
then $fl+2 \in \Sigma(eg,q).$

\emph{ Case 3}: $2 \nmid fl.$

\ (3.1) If $ (\mbox{odd},I-i) \equiv  \ \frac{e}{2} \  (\bmod ~2)$
for all $i \in I,$ then $fl +1 \in \Sigma(eg,q).$

\ (3.2) If $ i+ (\mbox{odd},I-i) \equiv \frac{e}{2} +|I_{\mbox{odd}}| \  (\bmod ~2)$
for all $i \in I,$ then $fl +1 \in \Sigma(g,q).$
\end{theorem}

\begin{proof}
\emph{ Case }1: For $2\mid f.$ Then $|\cS|=fl$ is even. By Lemma \ref{thm-three}
we have, for $a \in \cD_i, i \in I,$
\begin{eqnarray*}
\varphi_q(\Delta_{\cS}(a))  & \equiv &  i + (\mbox{odd}, I-i)\  (\bmod ~2)\\
\varphi_q(\Delta_{\widetilde{\cS}}(a))  & \equiv &   (\mbox{odd}, I-i) \  (\bmod ~2)\\
\varphi_q(\Delta_{\widetilde{\cS}}(0))  & \equiv &  0 \  (\bmod ~2).
\end{eqnarray*}

The conclusions (1.1) and (1.2) can be derived from Lemma \ref{thm-impor} (1).

\emph{ Case }2: For $2 \nmid f$ and $2 \mid l.$ Then $|\cS|=fl$ is even and for $a \in \cD_i, i \in I,$
\begin{eqnarray*}
\varphi_q(\Delta_{\cS}(a))  & \equiv &  i +\frac{e}{2} +(\mbox{odd}, I-i)\  (\bmod ~2)\\
\varphi_q(\Delta_{\widetilde{\cS}}(a))  & \equiv &  \frac{e}{2} + (\mbox{odd}, I-i) \  (\bmod ~2)\\
\varphi_q(\Delta_{\widetilde{\cS}}(0))  & \equiv &  |I_{\mbox{odd}}| \  (\bmod ~2).
\end{eqnarray*}

The conclusions (2.1) and (2.2) can be derived from Lemma \ref{thm-impor} (1).

\emph{ Case }3: For $2 \nmid fl,$ we have, for $a \in \cD_i, i \in I,$
\begin{eqnarray*}
\varphi_q(\Delta_{\cS}(a))  & \equiv &  (\mbox{odd}, I-i)\  (\bmod ~2)\\
\varphi_q(\Delta_{\widetilde{\cS}}(a))  & \equiv &  (\mbox{odd}, I-i)+i \  (\bmod ~2)\\
\varphi_q(\Delta_{\widetilde{\cS}}(0))  & \equiv &  \frac{e}{2} + |I_{\mbox{odd}}| \  (\bmod ~2).
\end{eqnarray*}

The conclusions (3.1) and (3.2) can be derived from Lemma \ref{thm-impor} (2).
\end{proof}

At the end of this section we show several general consequences of Theorem \ref{thm-six}.
For doing this we need to determine the parity of the number $(\mbox{odd}, I)$ for certain
subset $I$ of $\{ 0,1,\cdots, e-1\}.$

\begin{lemma}\label{thm-two}
Let $q=p^m \  (p \geq 3, m \geq 1), q-1=ef, 2 \mid e \geq 2,
(i,j)=(i,j)_e \ (i,j \in \Z_e =\Z / e\Z)$ be the $e$-th cyclotomic numbers on $\F_q.$
Then for $i \in \Z_e,$

(1) $(\mbox{odd},i)+(\mbox{even},i)=f-\delta_{i,0}.$

(2) Assume that $ 2 \mid f.$

If $ 2 \mid i,$ then $(\mbox{odd},i)=(\mbox{odd},-i), (\mbox{even},i)=(\mbox{even},-i).$

If $ 2 \nmid i,$ then $(\mbox{odd},i)=(\mbox{even},-i), (\mbox{even},i)=(\mbox{odd},-i).$ Particularly, if
$ e \equiv 2 \  (\bmod ~4),$ then $(\mbox{odd},\frac{e}{2})=(\mbox{even},\frac{e}{2})=\frac{f}{2}.$

(3)  Assume that $ 2 \nmid f.$

If $ 2 \mid i+\frac{e}{2},$ then $(\mbox{odd},i)=(\mbox{odd},-i), (\mbox{even},i)=(\mbox{even},-i).$

If $ 2 \nmid i+\frac{e}{2},$ then $(\mbox{odd},i)=(\mbox{even},-i).$ Particularly, if
$ e \equiv 2 \  (\bmod ~4),$ then $(\mbox{odd},0)=(\mbox{even},0)=\frac{f-1}{2}.$
\end{lemma}

\begin{proof}
(1) By Lemma \ref{thm-cyclotomy} (3),
$(\mbox{odd},i)+(\mbox{even},i)=\sum^{e-1}\limits_{j=0}(j,i)=f-\delta_{i,0}.$

(2) Assume that $ 2 \mid f.$ We have $(i,j)=(j,i).$ Then
$(\mbox{odd},i)=(i,\mbox{odd}),(\mbox{even},i)=(i,\mbox{even}).$
From $(i,j)=(-i,j-i)$ we get
\begin{eqnarray*}
(\mbox{odd},i)&=&(i,\mbox{odd})=\sum_{2 \nmid j}(i,j)=\sum_{2 \nmid j}(-i,j-i)=\sum_{2 \nmid j}(j-i,-i) \\
&=&
 \left \{
\begin{array}{ll}
(\mbox{odd},-i), & \mbox{if} \ 2 \mid i, \\
(\mbox{even},-i), &  \mbox{if} \ 2 \nmid i.
\end{array}
\right.
\end{eqnarray*}

Similarly,
$
(\mbox{even},i)=  \left \{
\begin{array}{ll}
(\mbox{even},-i), & \mbox{if} \ 2 \mid i, \\
(\mbox{odd},-i), &  \mbox{if} \ 2 \nmid i.
\end{array}
\right.
$

If $ e \equiv 2 \  (\bmod ~4),$ then $\frac{e}{2}$ is odd and
$(\mbox{odd},\frac{e}{2}) = (\mbox{even},-\frac{e}{2})=(\mbox{even},\frac{e}{2}).$
But $(\mbox{odd},\frac{e}{2}) + (\mbox{even},\frac{e}{2})=f,$ we get
$(\mbox{odd},\frac{e}{2}) = (\mbox{even},\frac{e}{2})=\frac{f}{2}.$

(3)  Assume that $ 2 \nmid f.$ From Lemma \ref{thm-cyclotomy}, we get
$$(j,i)=(i+\frac{e}{2}, j+\frac{e}{2})=(-(i+\frac{e}{2}),j-i)=(j-(i+\frac{e}{2}),-i).$$

Therefore, if $ 2 \nmid i+\frac{e}{2}, $ then
$$(\mbox{odd},i)=\sum_{2\nmid j}(j,i)=\sum_{2\nmid j}(j-(i+\frac{e}{2}),-i)=(\mbox{even},-i).$$

Similarly, if $ 2 \mid i+\frac{e}{2}, $ then
$(\mbox{odd},i)=(\mbox{odd},-i)$ and $(\mbox{even},i)=(\mbox{even},-i).$ If $ e \equiv 2 \  (\bmod ~4),$
then $ 2 \nmid \frac{e}{2}$ and $(\mbox{odd},0)=(\mbox{even},0).$ But
$(\mbox{odd},0) + (\mbox{even},0)=f-1.$ Therefore $(\mbox{odd},0) = (\mbox{even},0)=\frac{f-1}{2}.$
\end{proof}

\section{Examples}\label{sec-four}

After above preparation, now we show several results on the length of MDS self-dual codes as applications of
Theorem \ref{thm-six} and Lemma \ref{thm-three}. It is known that if
$ q \equiv 3 \  (\bmod ~4),$ and $ n \equiv 2 \  (\bmod ~4),$ then $n \notin \Sigma(q).$
Thus if $ q \equiv 3 \  (\bmod ~4),$ we consider the case $n=lf+a \ (a=0,1,2)$ with
$ n \not\equiv 2 \  (\bmod ~4).$ Firstly, we consider the case $l=|I|=1$ or $2.$

\begin{theorem}\label{thm-case1}
Let $q=p^m$ be a power of an odd prime $p, q-1=ef, 2 \mid e \geq 2.$

(1) If $2 \mid f,$ then $f  \in \Sigma(g,q).$ Moreover, if
$(\mbox{odd},0)$ is even, then $f+2 \in \Sigma(eg,q).$
Particularly, if $ e \equiv 2 \  (\bmod ~4)$ and $ f \equiv 0 \  (\bmod ~4),$ then
$f+2 \in \Sigma(eg,q).$

(2) If $2 \mid f,$ then $2f+2 \in \Sigma(eg,q).$ Moreover, if $4 \mid e,$ then $2f \in \Sigma(eg,q).$

(3) If $2 \nmid f.$  If
$(\mbox{odd},0)  \equiv   \frac{e}{2} \ (\bmod ~2),$  then $f+1 \in \Sigma(g,q)$ and  $f+1 \in \Sigma(eg,q).$
Particularly, if  $ e \equiv 2 \  (\bmod ~4)$ and $ f \equiv 3 \  (\bmod ~4),$ then
$f+1 \in \Sigma(g,q) \cap \Sigma(eg,q).$

(4) If $2 \nmid f$ and there exists $i$ such that
$1 \equiv i \equiv \frac{e}{2} \  (\bmod ~2)$ and $ (\mbox{odd},0) \equiv (\mbox{odd},i) \  (\bmod ~2),$
then $2f+2 \in \Sigma(eg,q).$
\end{theorem}

\begin{proof}
(1) Suppose that $2 \mid f.$ Take $I=\{0\},$ then $l=|I|=1$ and
$(\mbox{odd},I-0)=(\mbox{odd},0-0)=(\mbox{odd},0).$ From Theorem \ref{thm-six} (1.1) and (1.2),
we get $f \in \Sigma(g,q)$ and if $ 2 \mid (\mbox{odd},0),$ then $f+2 \in \Sigma(eg,q).$
Moreover, by Lemma \ref{thm-two} (1) and (3) we know that
$$
(\mbox{odd},i)=  \left \{
\begin{array}{ll}
(\mbox{odd},-i), & \mbox{if} \ 2 \mid i, \\
(\mbox{even},-i)=f-(\mbox{odd},-i) \equiv  (\mbox{odd},-i) \  (\bmod ~2), & \mbox{if} \ 2 \nmid i.
\end{array}
\right.
$$
Therefore
\begin{eqnarray*}
\sum^{e-1}_{i=0}(\mbox{odd},i) &=& (\mbox{odd},0)+(\mbox{odd},\frac{e}{2})+\sum^{\frac{e}{2}-1}_{i=1}((\mbox{odd},i)+(\mbox{odd},-i)) \\
& \equiv & (\mbox{odd},0)+(\mbox{odd},\frac{e}{2}) \  (\bmod ~2).
\end{eqnarray*}

On the other hand,
$$
\sum^{e-1}\limits_{i=0}(\mbox{odd},i)=\sharp \{ x \in \cD^{(2)}_1 : x+1 \neq 0 \}=  \left \{
\begin{array}{ll}
\frac{q-1}{2}, & \mbox{if} \  q \equiv 1 \  (\bmod ~4), \\
\frac{q-3}{2}, & \mbox{if} \ q \equiv  3 \  (\bmod ~4).
\end{array}
\right.
$$
If  $ e \equiv 2 \  (\bmod ~4)$ and $ f \equiv 0 \  (\bmod ~4),$ then
$ q \equiv 1 \  (\bmod ~8)$ and $  \sum^{e-1}\limits_{i=0}(\mbox{odd},i)=\frac{q-1}{2} \equiv 0 \  (\bmod ~2).$
Therefore $(\mbox{odd},0) \equiv (\mbox{odd},\frac{e}{2}) \  (\bmod ~2) .$ But from Lemma \ref{thm-two} (2),
$(\mbox{odd},\frac{e}{2})=\frac{f}{2} \equiv 0 \  (\bmod ~2),$ we get $2 \mid (\mbox{odd},0)$
and $f+2 \in \Sigma(eg,q).$

(2)  Assume that $2 \mid f.$ By Lemma \ref{thm-two} (2), we have
$$
(\mbox{odd},i)=  \left \{
\begin{array}{ll}
(\mbox{odd},-i), & \mbox{if} \ 2 \mid i, \\
(\mbox{even},-i)=f-(\mbox{odd},-i) \equiv  (\mbox{odd},-i) \  (\bmod ~2), & \mbox{if} \ 2 \nmid i.
\end{array}
\right.
$$

Therefore $(\mbox{odd},i) \equiv (\mbox{odd},-i) \  (\bmod ~2)$ for any $i.$ Then we get
\begin{eqnarray*}
\sum^{e-1}_{i=0}(\mbox{odd},i) &=& (\mbox{odd},0)+(\mbox{odd},\frac{e}{2})+\sum^{\frac{e}{2}-1}_{i=1}((\mbox{odd},i)+(\mbox{odd},-i)) \\
& \equiv & (\mbox{odd},0)+(\mbox{odd},\frac{e}{2}) \  (\bmod ~2).
\end{eqnarray*}
But $\sum\limits^{e-1}_{i=0}(\mbox{odd},i)=f \equiv  0 \  (\bmod ~2).$ Therefore
$(\mbox{odd},0) \equiv  (\mbox{odd},\frac{e}{2}) \  (\bmod ~2).$ Take $I=\{0,\frac{e}{2}\},$ then $ l=|I|=2,$
and $(\mbox{odd},I-0)=(\mbox{odd},I-\frac{e}{2})=(\mbox{odd},0)+(\mbox{odd},\frac{e}{2}) \equiv  0 \  (\bmod ~2).$
From Theorem \ref{thm-six} (1.2), we get $2f+2 \in \Sigma(eg,q).$ Moreover, if $ 4 \mid e,$ then $2f \in \Sigma(eg,q).$

(3) Suppose that $2 \nmid f,$ we also take $I=\{ 0 \}, l=1.$ From Theorem \ref{thm-six} case 3, we know that if
$(\mbox{odd},0) \equiv \frac{e}{2} \  (\bmod ~2),$ then $f+1 \in \Sigma(g,q)$ and  $f+1 \in \Sigma(eg,q).$
Moreover, if $ e \equiv 2 \  (\bmod ~4)$ and $ f \equiv 3 \  (\bmod ~4),$ then
$(\mbox{odd},0)=\frac{f-1}{2} \equiv 1 \equiv \frac{e}{2} \  (\bmod ~2),$
we get $f+1 \in \Sigma(g,q)$ and $f+1 \in \Sigma(eg,q).$

(4) Assume that $2 \nmid f.$ By Lemma \ref{thm-two} (3), we have
$ (\mbox{odd},i)=(\mbox{odd},-i)$ if $2 \mid i+\frac{e}{2}.$ If there exists
$i, 1 \equiv i \ \equiv \frac{e}{2} (\bmod ~2) $
such that $ (\mbox{odd},0) \equiv (\mbox{odd},2) \  (\bmod ~2),$ we take $I=\{0,i\},$
then $|I_{\mbox{odd}}| =1 \equiv \frac{e}{2} \  (\bmod ~2)$ and
$(\mbox{odd},I-0)=(\mbox{odd},0)+(\mbox{odd},i) \equiv 0  \equiv  (\mbox{odd},0)+(\mbox{odd},-i)
=(\mbox{odd},I-i)  \  (\bmod ~2).$ By Theorem \ref{thm-six} (2.2), we get $2f+2 \in \Sigma(eg,q).$
\end{proof}

Next we consider the semiprimitive case.

\begin{definition}
Let $p$ be a prime, $p \nmid e \geq 2. \ p$ is called semiprimitive module $e$ if
there exists a positive integer $t$ such that $p^t \equiv -1 \  (\bmod ~e).$

From now on, we take $t$ to be the least positive integer such that
$p^t \equiv -1 \  (\bmod ~e).$ Then the order of $p$ module $e$ is $2t.$
\end{definition}

\begin{lemma}
Let $ 2 \mid e \geq 4, p$ be a semiprimitive prime module $e$ and $t$ be the least
positive integer such that $p^t \equiv -1 \  (\bmod ~e).$
Let $r=p^m, q=r^2,m=ts, q-1=ef, R=r(-1)^s, \eta=\frac{R-1}{e}.$ Then,

(1) $2 \mid f$ and $\eta \in \gZ ;$

(2) Let $(i,j)=(i,j)_e \ (0 \leq i,j \leq e-1)$ be the cyclotomic numbers of order $e$
on $\F_q.$ Then $(\mbox{odd}, 0)$ is even and for $1 \leq i \leq e-1,$
$$
(\mbox{odd},i)=  \left \{
\begin{array}{ll}
\frac{R-1}{2} \  (\bmod ~2), & \mbox{if} \ 2 \mid i, \\
\frac{R-1}{2}+\eta  \ (\bmod ~2), & \mbox{if} \ 2 \nmid i.
\end{array}
\right.
$$
\end{lemma}

\begin{proof}
(1) $f=\frac{q-1}{e}=\frac{(r-1)(r+1)}{e}$ is even since $2 \nmid r$ and $r \equiv (-1)^s \  (\bmod ~e).$
Next, $R=p^{ts}(-1)^s \equiv (-1)^{s+s} \equiv 1 \  (\bmod ~e),$ we get $\eta=\frac{R-1}{e} \in \gZ.$

(2) For the semiprimitive case, the cyclotomic numbers have been determined in ([10], Lemma 5) as follows

$(0,0)=\eta^2 -(e-3)\eta -1, (0,i)=(i,0)=(i,i)=\eta^2 +\eta \ (1 \leq i \leq e-1)$,

$(i,j)=\eta^2 \ (1 \leq i \neq j \leq e-1).$

Then we get,
$$ (\mbox{odd}, 0) = \sum^{e-1}\limits_{\scriptstyle i=0
\atop \scriptstyle  2 \nmid i}(i,0)= \frac{e}{2}(\eta^2 +\eta)\equiv 0 (\bmod ~2),$$
and for $1 \leq i \leq e-1,$
$$
(\mbox{odd},i)= \sum^{e-1}\limits_{\scriptstyle j=1
\atop \scriptstyle  2 \nmid j}(j,i)=  \left \{
\begin{array}{ll}
\eta^2 \frac{e}{2} \equiv \eta \frac{e}{2} = \frac{R-1}{2} \  (\bmod ~2), & \mbox{if} \ 2 \mid i, \\
(i,i)+(\frac{e}{2}-1)\eta^2   \equiv  \frac{R-1}{2}+ \eta \ (\bmod ~2), & \mbox{if} \ 2 \nmid i.
\end{array}
\right.
$$
\end{proof}

\begin{theorem}\label{thm-semi}
Let $p$ be a semiprimitive prime module $e$ and $t$ be the least
positive integer such that $p^t \equiv -1 \  (\bmod ~e).$
Let $m=ts,r=p^m,q=r^2, q-1=ef, R=r(-1)^s, \eta=\frac{R-1}{e}.$

(1) If $2 \mid \eta,$ then $ lf \in \Sigma(g,q)$ for all $1 \leq l \leq \frac{e}{2}$
and $lf+2 \in \Sigma(eg,q)$ for all $1 \leq l \leq e.$

(2) If $2 \nmid \eta$ and $ 4 \mid e,$ then $ lf \in \Sigma(g,q)$ for `` all odd $l, 1 \leq l \leq e$ ''
and `` all even $l, 2 \leq l \leq  \frac{e}{2}$'', and $lf+2 \in \Sigma(eg,q)$ for `` all even $2 \leq l \leq e$ ''
and `` all odd $l,  1 \leq l \leq  \frac{e}{2}-1$''.

(3) If $2 \nmid \eta$ and $e \equiv 2 \  (\bmod ~4),$ then $fl \in \Sigma(g,q)$ for
`` all odd $l, \leq l \leq e-1$ '' and `` all even $2 \leq l \leq  \frac{e}{2}-1$'',
and $lf+2 \in \Sigma(eg,q)$ for `` all even $l,2 \leq l \leq e$ ''.
\end{theorem}

\begin{proof}Remark that $f$ is even (Lemma 9 (1)).

(1) If $ 2 \mid \eta =\frac{R-1}{e},$ then $2 \mid \frac{R-1}{2}$ and by Lemma 9,
$(\mbox{odd},i) \equiv 0 (\bmod ~2)$ for all $1 \leq i \leq e.$ For any $l,$ $1 \leq l \leq \frac{e}{2},$
we take a subset $I$ of $2 \gZ_e = \{0,2,4,\ldots,2e-2\}$ with size $|I|=l.$ For each $i \in I,$
$ i+ (\mbox{odd},I-i) \equiv \sum\limits_{j \in I}(\mbox{odd},j-i) \equiv 0 (\bmod ~2).$
By Theorem \ref{thm-six} (1.1)and (1.2),
we get $fl \in \Sigma(g,q)$ and $fl+2 \in \Sigma(eg,q)$ respectively.

(2) If $2 \nmid \eta=\frac{R-1}{e}$ and $4 \mid e,$ then $ 2 \mid \frac{R-1}{2}$ and for all $1 \leq i \leq e,$
$(\mbox{odd},i) \equiv i (\bmod ~2)$ (Lemma 9). Let $ 1 \leq l \leq e-1$ and $I$ be a subset of $\{1,2,\ldots,e\}$
with size $|I|=l.$ Then for each $i \in I,$
$$i+(\mbox{odd},I-i)=i+\sum_{j \in I}(\mbox{odd},j-i) \equiv i+\sum_{j \in I}(j+i) \equiv (l+1)i +\sum_{j \in I} j \ (\bmod ~2).$$
If $ 2 \nmid l,$ then $i+(\mbox{odd},I-i) \equiv \sum_{j \in I}j (\bmod ~2)$ are the same for all $i \in I.$ By Theorem \ref{thm-six} (1.1),
we get $ lf \in \Sigma(g,q).$ If $ 2 \mid l,$ we also take $I \subseteq 2 \gZ_e, |I|=l, 2 \leq l \leq \frac{e}{2},$ then for all $i \in I,$
$$i+(\mbox{odd},I-i)=i+\sum_{j \in I}j \equiv \sum_{j \in I}j  \ (\bmod ~2)$$
are the same for all $i \in I.$ By Theorem \ref{thm-six} (1,1), we get $ lf \in \Sigma(g,q).$

On the other hand, for each $I \subseteq \gZ_e, |I|=l,(\mbox{odd},I-i) \equiv li+ \sum_{j \in I}j (\bmod ~2).$
If $2 \mid l, 2 \leq l \leq e,$ it is easy to see that we have a subset $I$ of $\gZ_e$ such that $|I|=l$ and
$\sum\limits_{j \in I}j \equiv 0 \ (\bmod ~2).$ Therefore $(\mbox{odd},I-i) \equiv 0 \  (\bmod ~2)$ for all $i \in I.$
By Theorem \ref{thm-six} (1.2), we get $lf+2 \in \Sigma(eg,q).$

If $ 2 \nmid l$ and $1 \leq l \leq \frac{e}{2}-1, $ we take a subset $I$ of $2\gZ_e$ with size $|I|=l.$
We also have $(\mbox{odd},I-i) \equiv i+\sum_{j \in I}j \equiv 0 \  (\bmod ~2)$ for all $i \in I.$ Then we get
$lf+2 \in \Sigma(eg,q)$ by Theorem \ref{thm-six} (1.2).

(3) If $2 \nmid \eta=\frac{R-1}{e}$ and $ e \equiv 2 \  (\bmod ~4),$ then $ 2 \nmid \frac{R-1}{2}.$
By Lemma 9, $ 2 \mid (\mbox{odd},0)$ and $ (\mbox{odd},i) \equiv i+1 \  (\bmod ~2)$ for $1 \leq i \leq e-1.$
Let $I$ be a subset of $\gZ_e$ with size $|I|=l,1 \leq l \leq e.$ Then
$$
i+\sum_{j \in I}(\mbox{odd},J-I)\equiv     \left \{
\begin{array}{ll}
\sum\limits_{\scriptstyle j \in I
\atop \scriptstyle  2 \nmid j} 1\  (\bmod ~2), & \mbox{if} \ 2 \mid i, \\
1+\sum\limits_{\scriptstyle j \in I
\atop \scriptstyle  2 \mid j}1 \equiv 1+l +\sum\limits_{\scriptstyle j \in I
\atop \scriptstyle  2 \nmid j}1  \  (\bmod ~2), & \mbox{if} \ 2 \nmid i.
\end{array}
\right.
$$
By Theorem \ref{thm-six} (1.1), we get $ lf \in \Sigma(g,q) $ for `` odd $l,1 \leq l \leq e$",
and `` even $l, 2 \leq l \leq \frac{e}{2}-1$". On the other hand, for $i \in I,$
$$
\sum_{i \in I}(\mbox{odd},I-i) \equiv
\sum\limits_{\scriptstyle j \in I
\atop \scriptstyle  j \neq i, 2 \nmid j-i} 1 \equiv   \left \{
\begin{array}{ll}
 -1 +A \  (\bmod ~2), & \mbox{if} \ 2 \mid i, \\
-1+B  \  (\bmod ~2), & \mbox{if} \ 2 \nmid i.
\end{array}
\right.
$$
where $A=\sharp \{ j \in I : 2 \mid j\},B=\sharp \{ j \in I : 2 \nmid j\}, |A+B|=|I|=l. $
Let $2 \mid l$ and $2 \leq l \leq e.$ It is easy to find a subset $I$ of $\gZ_e$ with size $|I|=l$
such that both of $A$ and $B$ are odd. Then for each $i \in I,$
$$\sum_{j \in I}(\mbox{odd},I-i) \equiv A-1 \ \mbox{or} \ B-1 \equiv 0 \  (\bmod ~2).$$
By Theorem \ref{thm-six} (1.2), $lf+2 \in \Sigma(eg,q)$ for all even $l, 2 \leq l \leq e.$
\end{proof}

\section{Results for the cases $e=2$ and $4$}\label{sec-cyclotomy}

\subsection{When $e=2$}

Let $q=p^n, p \geq 3, q-1=2f.$ The cyclotomic numbers of order 2 are given in Lemma \ref{thm-cyc}.

\begin{theorem}\label{thm-e2}
Let $q=p^n, p \geq 3, f=\frac{q-1}{2}.$

(1) If $ q \equiv 1 \  (\bmod ~4)$, then $ 2f+2 \in \Sigma(eg,q).$
Moreover, if $ q \equiv 1 \  (\bmod ~8),$ then $f+2 \in \Sigma(eg,q).$

(2) If $ q \equiv 3 \  (\bmod ~4)$, then $2f+2 \in \Sigma(eg,q).$
Moreover, if $ q \equiv 7 \  (\bmod ~8),$ then $f+1 \in \Sigma(eg,q) \cap  \Sigma(g,q) .$
\end{theorem}

\begin{proof}
(1) If $ q \equiv 1 \  (\bmod ~4)$, then $ 2 \mid f.$ Take $I=\{ 0,1 \}, $
$$(\mbox{odd},I-0)=(\mbox{odd},I-1)=(\mbox{odd},0)+(\mbox{odd},1)=(1,0)+(1,1) \equiv 0 \  (\bmod ~2),$$
from Theorem \ref{thm-six} (1.2), we get $2f+2 \in \Sigma(eg,q).$
Moreover, if $ q \equiv 1 \  (\bmod ~8),$ we take $I=\{ 0 \}$ or $\{ 1 \},$
$$(\mbox{odd},I-0)= (\mbox{odd},0)=(1,0)=\frac{f}{2} \ \mbox{or} \  (\mbox{odd},I-1)=(\mbox{odd},0)=(1,0)=\frac{f}{2}$$
are even, so $f+2 \in \Sigma(eg,q).$

(2) If $ q \equiv 3 \  (\bmod ~4)$, then $ 2 \nmid f.$ Take $I=\{ 0,1 \}, $ then
 $ |I_{\mbox{odd}}|=1=\frac{e}{2}$ and $(\mbox{odd},I-0)=(\mbox{odd},I-1)=(1,0)+(1,1)=f-1 \equiv 0 \  (\bmod ~2),$
from Theorem \ref{thm-six} (2.2), we get $2f+2 \in \Sigma(eg,q).$ Moreover, if $ q \equiv 7 \  (\bmod ~8),$
we take $I=\{ 0 \}$
$$(\mbox{odd},I-0)= (\mbox{odd},0)=(1,0)=\frac{f-1}{2} \equiv 1 \  (\bmod ~2),$$
from Theorem \ref{thm-six} (3.1), we get $f+1 \in \Sigma(eg,q).$

If we take $I=\{ 1 \}$,
$$ 1 + (\mbox{odd},I-1)= 1+ (\mbox{odd},0)=1+(1,0)=\frac{f+1}{2} \equiv  \frac{e}{2}+ |I_{\mbox{odd}}| \  (\bmod ~2),$$
from Theorem \ref{thm-six} (3.2), we get $f+1 \in \Sigma(g,q).$
\end{proof}

\subsection{When $e=4$}

The cyclotomic numbers of order 4 are given in Lemma \ref{thm-cyc}.
Now we get the following constructions of self-dual MDS codes
by  using cyclotomic classes of order four.

\begin{theorem}
Let $q=p^n   \equiv 1\  (\bmod ~4), q-1=4f.$

(1) Assume that $ p \equiv 1\  (\bmod ~4).$

(1.1) If $2 \mid f$ (namely $q \equiv 1\  (\bmod ~8)$), then $f,2f \in \Sigma(g,q)$
and $f+2,2f+2,4f+2 \in \Sigma(eg,q).$ Moreover, if $q \equiv 1\  (\bmod ~16)$ and $t \equiv 2 \  (\bmod ~4),$
or $q \equiv 9\  (\bmod ~16)$ and $ 4 \mid t ,$ then $3f \in \Sigma(g,q).$
If $q \equiv 1 \  (\bmod ~16)$ and $4 \mid t,$
or $q \equiv 9\  (\bmod ~16)$ and $t \equiv 2\  (\bmod ~4),$ then $3f+2 \in \Sigma(eg,q).$

(1.2) If $2 \nmid f$ (namely $q \equiv 5\  (\bmod ~8)$), then $2f \in \Sigma(g,q)$
and $f+1,4f+2 \in \Sigma(eg,q).$

(2) Assume that $ p \equiv 3\  (\bmod ~4).$ Then $n=2m$ is even.

(2.1) If $ p \equiv 7\  (\bmod ~8)$ or $ p \equiv 3\  (\bmod ~8)$ and $2 \mid m,$
then $f,2f \in \Sigma(g,q)$ and $fl+2 \in \Sigma(eg,q)$ for $ 1 \leq l\leq 4.$

(2.2) If $ p \equiv 3\  (\bmod ~8)$ and $2 \nmid m,$
then $fl \in \Sigma(g,q)$ for $ 1 \leq l\leq 3$ and $fl+2 \in \Sigma(g,q)$
for $l=1,2,4.$
\end{theorem}

\begin{proof}
(1.1) If $2 \mid f$, the conclusion $f,2f \in \Sigma(g,q)$
can be derived from Theorem \ref{thm-six} (1.1) and Lemma \ref{thm-cyc} (2.1) by taking
$I=\{0\} $ and $\{0,2\}$  respectively.
The conclusion
 $f+2,2f+2,4f+2 \in \Sigma(eg,q)$ can be derived from Theorem \ref{thm-six} (1.2)
 and Lemma \ref{thm-cyc} (2.1) by taking
$I=\{0\}, \{0,2\}$ and $\{ 0,1,2,3\}$ respectively.

Moreover, from
$q \equiv 1\  (\bmod ~8), q=s^2 +4t^2$ and $ s \equiv 1 \  (\bmod ~4),$
we get $ 2 \mid t.$ If $ q \equiv 1 \  (\bmod ~16), 4 \mid t$ or
$ q \equiv 9\  (\bmod ~16),  t \equiv 2\  (\bmod ~4),$
then $ 8(\mbox{odd},1)=q-1-4t  \equiv 0\  (\bmod ~16).$
By Lemma \ref{thm-cyc}, $(\mbox{odd},i) \equiv  0 \  (\bmod ~2)$ for all $ 0 \leq i \leq 3.$
From Theorem \ref{thm-six} (1,1) and (1,2), we get $3f \in \Sigma(g,q),
3f+2 \in \Sigma(eg,q)$ by taking $I=\{0,1,2\}$ respectively.

(1.2) If $2 \nmid f,$ the conclusion $2f \in \Sigma(g,q)$
can be derived from Theorem \ref{thm-six} (II,1) and Lemma \ref{thm-cyc} (2.2), by taking $I=\{0,1\}$
or $\{0,3\}$ provided $ (\mbox{odd},1)=0$  or $ (\mbox{odd},3)=0$ respectively. The conclusion
$4f+2 \in \Sigma(eg,q)$ can be derived from Theorem \ref{thm-six} (II,2) and Lemma \ref{thm-cyc}, by taking
$I=\{0,1,2,3\}.$ The conclusion $f+1  \in \Sigma(eg,q)$ can be derived from Theorem \ref{thm-six}
by taking $I=\{0\}.$

(2) Assume that $q \equiv 3  \equiv -1\  (\bmod ~4).$ This is the semiprimitive case. Thus
$q=r^2, r=p^m, 2 \mid f=\frac{q-1}{4}$ and $\eta=\frac{(-p)^m-1}{4}.$ If $ p \equiv 7\  (\bmod ~8)$
or $ p \equiv 3\  (\bmod ~8)$ and $2 \mid m,$ then $2 \mid \eta.$ If
$ p \equiv 3\  (\bmod ~8)$ and $2 \nmid m,$ then $2 \nmid \eta.$ The conclusion of (2.1) and (2.2)
can be derived directly from Theorem \ref{thm-semi}.
\end{proof}

\begin{remark} We have examples satisfying the conditions provided in (1.1). Let $q=p^n, p \equiv 1\  (\bmod ~4).$
For condition $q \equiv 1\  (\bmod ~16)$ and $4 \mid t,$ we have example $q=113=s^2 +4t^2=(-7)^2 +4 \cdot 4^2.$
For condition $q \equiv 9 \  (\bmod ~16)$ and $t \equiv 2\  (\bmod ~4),$ we have examples
$q=25=(-3)^2 +4 \cdot 2^2$ and $q=41=5^2 +4 \cdot 2^2.$
\end{remark}

\begin{remark}
One of our further work is consider the existence of MDS self-dual codes via cyclotomic numbers of order 6 and 8.
\end{remark}

\section*{Acknowledgements}

  A. Zhang's research was supported by the Natural Science Foundation of China under Grant No:11401468
 and the Natural Science Basic Research Plan in Shaanxi Province of China under Grant No:2019JQ-333.


\end{document}